\documentclass[
amsmath,amssymb]{revtex4}
\usepackage{epsfig}
\usepackage{amsthm}
\usepackage{amssymb}
\usepackage{latexsym}
\usepackage{url}
\usepackage{amsmath, amscd}
\usepackage{verbatim}

\newtheorem{thm}{Theorem}

\theoremstyle{definition}

\newcommand{\R}{\mathbb{R}}

\begin{document}
\title{A possible mathematics for the unification of quantum mechanics and general relativity}
\author{A. \surname{Kryukov}} 
\affiliation{Department of Mathematics, University of Wisconsin Colleges, 780 Regent Street, Madison, Wisconsin 53708, USA}
\begin{abstract}

The paper summarizes and generalizes a recently proposed mathematical framework that unifies the standard formalisms of special relativity and quantum mechanics. The framework is based on Hilbert spaces $H$ of functions of four space-time variables ${\bf x}, t$, furnished with an additional indefinite inner product invariant under Poincar{\'e} transformations, and isomorphisms of these spaces that preserve the indefinite metric. The indefinite metric is responsible for breaking the symmetry between space and time variables and for selecting a family of Hilbert subspaces that are preserved under Galileo transformations. Within these subspaces the usual quantum mechanics with Shr{\"o}dinger evolution and $t$ as the evolution parameter is derived. Simultaneously, the Minkowski space-time is isometrically embedded into $H$, Poincar{\'e} transformations have unique extensions to isomorphisms of $H$ and the embedding commutes with Poincar{\'e} transformations. The main new result is a proof that the framework accommodates arbitrary pseudo-Riemannian space-times furnished with the action of the diffeomorphism group.   
\end{abstract}


\maketitle


\section{Preliminaries}

Resolving the tension between quantum theory and relativity is one of the most important problems of modern physics. In this paper the issue is presented mathematically and the goal is to find a formalism that unifies the mathematics used in both theories. More specifically, the challenge is to find a meaningful way of ``encoding'' the differential geometry of finite dimensional pseudo-Riemannian manifolds into the theory of infinite dimensional Hilbert spaces.
Provided such a ``unified'' formalism exists, it may be useful in particular, in studying the issues of compatibility of quantum theory and relativity and the problem of emergence of the classical world in quantum theory. 

The first thing that comes to mind is that the theory of representations of groups provides a partial answer to the challenge. Indeed, it allows one to represent symmetries of a physical system in terms of linear transformations on the Hilbert space of states of the system. For instance, symmetries of Minkowski space-time can be represented by unitary transformations on the Hilbert space of states of a relativistic system. However, despite the undeniable significance of representations in physics there is more to the problem than representations of groups can provide. For instance, if the symmetry group (i.e., group of isometries) of a Riemannian manifold is trivial, representations of the group do not contain any information about the manifold. 

On the other hand, there is a link between the topology of a space and algebra of continuous functions on the space that may be useful to tackle the problem. Namely, the
celebrated Gel'fand-Kolmogorov theorem Ref.\cite{Gel} states that an arbitrary compact Hausdorff space $X$ is homeomorphic to the space of all evaluation homomorphisms (delta functions) in the infinite-dimensional vector space dual to the Banach algebra $C(X)$ of continuous functions on $X$. In other words, $X$ can be identified with the set of all delta functions in the space dual to $C(X)$. However, $C(X)$ is not a Hilbert space, the topology is poorer than the Riemannian structure and the condition of compactness is too restrictive. In addition, the fact that elements of $C(X)$ are functions on $X$ makes it difficult to use $C(X)$ independently of $X$. This is a problem if one has in mind the goal of deriving the classical from the quantum theory.

In the case of a single particle system in $\R^{3}$ the most obvious physically meaningful embedding of $\R^{3}$ into the space of states of a particle is by identifying a point ${\bf a}$ in $\R^{3}$ with the state $\delta^{3}_{\bf a}({\bf x})=\delta^{3}({\bf x}-{\bf a})$ of the particle found at ${\bf a}$. This embedding was usefully explored in Refs.\cite{Kryukov},\cite{Kryukov3} to develop a geometric approach to quantum mechanics. Results of Ref.\cite{KryukovJMP} prove that an embedding of this kind is also ideally suited for addressing the issues of the unification of quantum mechanics and special relativity. In the present paper the main results of Ref.\cite{KryukovJMP} are summarized, updated and extended to include curved space-time manifolds of general relativity.

\section{Quantum mechanics and special relativity}

Delta functions are not in the common space $L_{2}(\R^{3})$ of Lebesgue square-integrable functions on $\R^{3}$. So to use this correspondence one first needs to find a Hilbert space of functions that contains delta functions and that is ``approximately equal'' to the space $L_{2}(\R^{3})$. The following theorem takes care of this task (see Refs.\cite{Kryukov3},\cite{KryukovJMP} for details on results in this section and Refs.\cite{Gel-Kos},\cite{Gross} for related original publications on rigged Hilbert spaces).
\begin{thm}
\label{1}
The Hilbert space ${\bf H}$ obtained by completing the space $L_{2}(\R^{3})$ in the metric defined by the inner product 
\begin{equation}
\label{inner}
(\varphi, \psi)_{{\bf H}}=\left(\frac{L}{\sqrt {2\pi}}\right)^{3}\int e^{-\frac{L^{2}}{2}({\bf x}-{\bf y})^{2}}\varphi({\bf x}){\overline \psi({\bf y})} d^{3}{\bf x}d^{3}{\bf y}
\end{equation}
with a positive constant $L$ contains delta functions and their derivatives. Furthermore, for a sufficiently large $L$ the ${\bf H}$ and $L_{2}$-norms of any given function $f\in L_{2}(\R^{3})$ are arbitrarily close to each other.
\end{thm}
The map $\omega:a \longrightarrow \delta^{3}_{\bf a}$ is one-to-one, so the set $\R^{3}$ can be identified with the set of all delta functions in ${\bf H}$. Moreover, the induced manifold structure and the metric on the image $M_{3}=\omega(\R^{3})\subset {\bf H}$ are those of the Euclidean space $\R^{3}$. In other words, 
\begin{thm}
\label{1a}
The map $\omega:{\bf a} \longrightarrow \delta^{3}_{\bf a}$ is an isometric embedding of the space $\R^{3}$ with the Euclidean metric into the space ${\bf H}$ defined in theorem \ref{1}. 
\end{thm}
Note that the map $\omega$ is not linear. In particular, because  the norm of any delta function $\delta^{3}_{\bf a}$ in ${\bf H}$ is the same, the image $\omega(\R^{3})$ is a submanifold of the sphere in ${\bf H}$. So the vector structure on $\R^{3}$ is not compatible with the vector structure on ${\bf H}$. However, one can introduce a vector structure on the image $M_{3}$ by defining the operations of addition $\oplus$ and multiplication by a scalar $\lambda \odot$ via $\omega(a)\oplus\omega(b)=\omega(a+b)$ and $\lambda \odot\omega(a)=\omega(\lambda a)$. Moreover, because $\omega$ is a homeomorphism onto $M_{3}$, these operations are continuous in the topology of $M_{3}\subset {\bf H}$.

The Riemannian structure of the Euclidean space is now ``encoded'' into the Hilbert space ${\bf H}$. 
At the same time, the space ${\bf H}$ is approximately equal to the space $L_{2}(\R^{3})$ (symbolically, ${\bf H}\approx L_{2}(\R^{3})$). That is, provided the constant $L$ in theorem \ref{1} is sufficiently large, the ${\bf H}$-norms of typical square-integrable functions will be as close as we wish to their $L_{2}(\R^{3})$ norms.
Accordingly, if ${\bf H}$ is used in place of $L_{2}(\R^{3})$ in the usual quantum mechanics, the expected values, probabilities of transition and other measured quantities remain practically the same, ensuring consistency with experiment. In the following the constant $L$ will be set to one and the needed agreement between spaces ${\bf H}$ and $L_{2}(\R^{3})$ will be achieved by an appropriate choice of units.

The next step is to extend results of theorems \ref{1}, \ref{1a} to Minkowski space-time.
For this one needs to work with spaces of functions of four variables ${\bf x}, t$.  Let ${\widetilde H}$ be the Hilbert space of functions $f$ of four variables $x=({\bf x},t)$ that is the completion of the space $L_{2}(\R^{4})$ in the metric given by the kernel $e^{-\frac{1}{2}(x-y)^{2}}$. It is easy to see that ${\widetilde H}$ is the orthogonal sum of the subspace ${\widetilde H}_{\rm ev}$ of all functions that are even in the time variable $t$ and the subspace ${\widetilde H}_{\rm odd}$ of all functions that are odd in $t$.
The following theorem generalizes the results of theorem \ref{1} to the case of Minkowski space-time.
\begin{thm}
\label{2}
Let $H$ be the set of all functions $f({\bf x},t)=e^{-t^{2}}\varphi({\bf x},t)$ with $\varphi \in \widetilde{H}$. Consider the Hermitian form $(f,g)_{H_{\eta}}$ on $H$ given by
\begin{equation}
\label{innerM}
(f,g)_{H_{\eta}}=\int e^{-\frac{1}{2}({\bf x}-{\bf y})^{2}+\frac{1}{2}(t-s)^{2}}f({\bf x},t){\overline g({\bf y},s)} d^{3}{\bf x}dtd^{3}{\bf y}ds
\end{equation}
and let $(f,f)_{H_{\eta}}\equiv \left\|f\right\|^{2}_{H_{\eta}}$ be the corresponding quadratic form, or the squared $H_{\eta}$-norm.
Then $H$ is exactly the set of functions whose even and odd components have a finite $H_{\eta}$-norm.
Moreover, $H$ furnished with the inner product $(f,g)_{H_{+}}=(\varphi, \psi)_{\widetilde{H}}$, where $f({\bf x},t)=e^{-t^{2}}\varphi({\bf x},t)$, $g({\bf x},t)=e^{-t^{2}}\psi({\bf x},t)$ is a Hilbert space. 
The Hermitian form (\ref{innerM}) defines an indefinite, non-degenerate inner product on $H$, such that $\left\|f\right\|^{2}_{H_{\eta}}>0$ for all even functions $f\neq 0$ and  $\left\|f\right\|^{2}_{H_{\eta}}<0$ for all odd functions $f \neq 0$ in $H$. Finally, $H$ contains the delta functions $\delta^{4}_{a}(x)=\delta^{4}(x-a)$ and their derivatives. 
\end{thm}
The space $H$ is an example of what is called the {\em Krein space}. A Krein space is a complex vector space $V$ with a Hermitian inner product $(f,g)_{V}$ and such that $V$ is the direct sum of two spaces $H_{1}, H_{2}$ that are Hilbert with respect to the inner products $(f,g)_{V}$ and $-(f,g)_{V}$ respectively and that are orthogonal in the inner product on $V$.  The following analogue of theorem \ref{1a} is valid:
\begin{thm}
\label{3} 
The map $\omega: N \longrightarrow H$, $\omega(a)=\delta^{4}_{a}$ is an embedding that identifies the Minkowski space $N$ with the submanifold $M_{4}$ of $H$ of all delta functions $\delta^{4}_{a}$, $a \in N$. Under the embedding the indefinite metric on $H$ yields the Minkowski metric on $M_{4}$, while the ${\widetilde H}$-metric yields the ordinary Euclidean metric on $M_{4}$.  
\end{thm}
So, similarly to the space $\R^{3}$, the Minkowski space $N$ is now encoded into the space $H$. As before, the map $\omega$ is not linear, but the image $M_{4}$ can be furnished with a linear structure, induced from $N$. 

To ``lift'' the theory of relativity from $N$ onto $H$ it turns out to be important that the set $M_{4}$ is a complete set in $H$ (i.e., no element of $H$ is orthogonal to all delta functions in $M_{4}$) and that elements of any finite subset of $M_{4}$ are linearly independent. In this sense, the set $M_{4}$ forms a basis of the space $H$. Because  of that physics on $M_{4}$ obtains a unique ``linear extension'' to the entire Hilbert space $H$.
If $\Pi$ is a Poincar{\'e} transformation, and $f$ is a function(al) in $H$, then $\delta_{\Pi}: f \longrightarrow f\circ \Pi^{-1}$ is a linear map on $H$. Because the Hilbert metric on $H$ is not invariant under general Poincar{\'e} transformations, the operator $\delta_{\Pi}$ may not be bounded as a map into $H$ so that the map  $\Pi \longrightarrow \delta_{\Pi}$ is not a representation of the Poincar{\'e} group $P$. However, the set of all functions $f\circ \Pi^{-1}$ with a fixed $\Pi$ and $f \in H$ form a Hilbert space $H'$ with the inner product defined by $(\delta_{\Pi}f, \delta_{\Pi}g)_{H'_{+}}=(f,g)_{H_{+}}$. The map $\delta_{\Pi}: H \longrightarrow H'$ is then an isomorphism of Hilbert spaces. Hilbert spaces obtained in such a way can be thought of as different realizations of one and the same abstract Hilbert space ${\bf S}$.
A particular isomorphism $\Gamma: {\bf S} \longrightarrow H$ can be thought of as a coordinate chart on ${\bf S}$ Ref.\cite{Kryukov3}. With this in mind one can formulate the following essential result. In the theorem the expression {\em isometric embedding} refers to an embedding that preserves the indefinite metric. Likewise {\em isomorphism} is an isomorphism of Hilbert spaces that in addition preserves the indefinite metric.
\begin{thm}
\label{5}
Let $\Gamma:{\bf S}\longrightarrow H$ be an isomorphism of the abstract Hilbert space ${\bf S}$ with an additional indefinite metric onto the space $H$ of functions defined in theorem \ref{2}. Let $\gamma: N \longrightarrow \R^{1,3}$ be a global coordinate chart from the Minkowski space-time onto the coordinate space of the observer in an inertial reference frame $K$. Let $\Pi$ be a Poincar{\'e} transformation that relates coordinates associated with frames $K$ and $K'$. Then there exists a unique isometric embedding $\Omega$ and a unique isomorphism $\delta_{\Pi}:H \longrightarrow H'$ such that the diagram
\begin{equation}
\label{diagram}
\begin{CD}
{\bf S}   @ >\Gamma>> H @ >\delta_{\Pi}>> H'\\
@ AA\Omega A           @ AA \omega A      @ AA\omega A \\
N         @ >\gamma>> \R^{1,3} @ >\Pi>> \R^{1,3}
\end{CD}
\end{equation}
\newline
is commutative. It follows that within the assumptions of the theorem the embedding $\omega$ preserves the structure of special relativity and extends it in a unique way to the abstract Hilbert space ${\bf S}$.  
\end{thm}
A couple of remarks.
\begin{enumerate}
\item
Note that $\delta_{\Pi}$ maps delta functions to delta functions, so that in accord with the diagram (\ref{diagram}) the image of the manifold $M_{4}$ in $H$ is the submanifold $M_{4}$ in $H'$. This together with the fact that $\omega$ is an isometric embedding is what allows for the usual theory of relativity on Minkowski space-time to be a part of the new framework. At the same time completeness of the set $M_{4}$ together with linear independence of its elements makes the entire construction rigid, ensuring uniqueness of the extension. 
\item 
The proposed method of extension of the Poincar{\'e} group action from $N$ onto ${\bf S}$ can be applied to {\em any} group acting on $N$. 
Notice that an arbitrary non-linear transformation acting continuously on $N$ becomes linear when extended to ${\bf S}$. That is so because moving across $N$ corresponds to going ``across dimensions'' of $H$ so that a linear extension of the transformation becomes possible. Completeness of the set $M_{4}$ in $H$ ensures then that such an extension is unique. 
\item 
 In explicit terms the  covariance of the construction amounts to the following: {\em(a)} the embedding preserves covariant properties of $4$-tensors (elements of the tensor algebra of Minkowski space-time); {\em(b)} the involved functional objects are also covariant under Poincar{\'e} transformations; {\em(c)} the embedding is equivariant, that is, it commutes with the action of the Poincar{\'e} group. 
 The first two properties simply mean that the usual $4$-tensors are also elements of the tensor algebra of ${\bf S}$ and that all considered objects are tensorial. The third property signifies that Poncar{\'e} transformations $\Pi \in P$ can be identified with morphisms $\delta_{\Pi}$ of Hilbert spaces.
 All three properties follow from the diagram (\ref{diagram}). Indeed, because $\omega$ is an embedding, the differential map $d\omega$ yields embedding of the corresponding tangent and, more generally, tensor bundles, which proves {\em(a)}. Property {\em(c)} is exactly the commutative property of the diagram. To prove {\em(b)} note that
a function $f \in H$ represents an invariant element of ${\bf S}$ and transforms as a vector under $\delta_{\Pi}$: $f'=\delta_{\Pi} f$. 
Writing the law $f'=\delta_{\Pi} f$ in the form $f'(x')=(\delta_{\Pi} f)(x')=f(\Pi^{-1} x')=f(x)$, one recovers the usual law of transformation of scalar functions. 
 The metric operators ${\widehat G}_{H_{+}},{\widehat G}_{H_{\eta}}:H \longrightarrow H^{\ast}$, where $H^{\ast}$ is the dual of $H$ and $(f,g)_{H_{+}}=({\widehat G}_{H_{+}}f,g)$, $(f,g)_{H_{\eta}}=({\widehat G}_{H_{\eta}}f,g)$, define $2$-forms in the tensor algebra of ${\bf S}$. Their transformation law ${\widehat G}_{H'_{+}}=\delta^{\ast -1}_{\Pi}{\widehat G}_{H_{+}}\delta^{-1}_{\Pi}$ and ${\widehat G}_{H'_{\eta}}=\delta^{\ast -1}_{\Pi}{\widehat G}_{H_{\eta}}\delta^{-1}_{\Pi}$, where $\delta^{\ast}_{\Pi}:H'^{\ast}\longrightarrow H^{\ast}$ is the adjoint of $\delta_{\Pi}$, ensures invariance of the inner products.  It follows that the metric operators are also covariant quantities. 
 Note that although no covariant equations for functional quantities were considered so far, it is clear that they simply are tensor equations for fields with values in the tensor algebra of ${\bf S}$ (see Ref.\cite{Kryukov3}). 
 \item 
Let us call the realization $\Gamma:{\bf S} \longrightarrow H$ of ${\bf S}$ a $K$-representation.
The diagram (\ref{diagram}) demonstrates that under the transformation of the frame $K$ by $\gamma \longrightarrow \Pi \circ \gamma$ the $K$-representation changes in a covariant fashion to a unitary equivalent realization $\delta_{\Pi}\circ \Gamma: {\bf S} \longrightarrow H'$ of ${\bf S}$. According to the diagram, this realization is a unique extension of the coordinate system $\Pi \circ \gamma$ of an observer in the reference frame $K'$, or the {\em $K'$-representation of ${\bf S}$}. Note that in general the spaces $H$ and $H'$ have a different functional content. However, both spaces are realizations of the same invariant abstract Hilbert space ${\bf S}$ with the invariant Hilbert and indefinite metrics on it. In other words, only a functional realization of ${\bf S}$ changes from frame to frame, not the space ${\bf S}$ itself. For applications to physics it is particularly important that the inner products of elements of ${\bf S}$ in all realizations remain the same.
In the following it will be advocated that ${\bf S}$ is an appropriate physical space of states of a quantum system. Suppose for now that this is the case and consider an observer in an arbitrary inertial frame $K'$ having access to the space ${\bf S}$ and describing it via $K'$-representation. Then the observer will not be able to use the functional content of the Hilbert space $H'$ of representation or the representation itself to determine the state of motion of the frame $K'$. Rather, similar to the ordinary special relativity, a particular functional realization of the space ${\bf S}$ is not physical, i.e., all such realizations are physically equivalent.  
\item The Poincar{\'e} transformations $\Pi$ and their extensions $\delta \Pi$ in the theorem are ``passive'' transformations i.e., they describe changes in coordinate realizations of the fixed, invariant spaces $N$ and ${\bf S}$. The corresponding ``active'' version of the theorem is also possible and is given by the following analogue of diagram (\ref{diagram}):
\begin{equation}
\label{diagramA}
\begin{CD}
{\bf S'}  @ <\delta_{\Pi} << {\bf S} @ >\Gamma>> H\\
@ AA\Omega A           @ AA \Omega A      @ AA\omega A \\
N        @  <\Pi<<   N @ >\gamma>> \R^{1,3}
\end{CD}
\end{equation}
Here the maps $\gamma, \Gamma, \omega$ and the embedding $\Omega:N \longrightarrow {\bf S}$ are the same as before. The Poincar{\'e} transformation $\Pi$ maps Minkowski space $N$ onto itself.  The space ${\bf S'}$ contains the subset $\Omega(N)$ as a complete set and is otherwise defined by the diagram. The isomorphism $\delta_{\Pi}$ is the linear extension of the map $\Omega \Pi \Omega^{-1}$ from $\Omega(N)$ onto ${\bf S}$. 
\end{enumerate}

Let's now turn to the embedding of quantum mechanics into the same framework. The first thing to do is to relate the Hilbert space $H$ of functions of four variables ${\bf x}, t$ to the usual Hilbert spaces of functions of three variables ${\bf x}$ with $t$ as a parameter of evolution. For this consider the family of subspaces $H_{\tau}$ of $H$ each consisting of all functionals $\varphi_{\tau}({\bf x},t)=\psi({\bf x},t)\delta(t-\tau)$ for some fixed $\tau \in \R$. 
\begin{thm}
\label{6}
Under the inclusion $i:H_{\tau}\longrightarrow H$ the indefinite inner product on $H$ yields a Hilbert metric on $H_{\tau}$ for all $\tau \in \R$. 
Furthermore,
let ${\bf H} \approx L_{2}(\R^{3})$ be the Hilbert space  defined in theorem \ref{1} (with $L=1$ and a sufficiently small scale to make the approximation valid). 
Then for all $\tau \in \R$ the map $I:H_{\tau}\longrightarrow {\bf H}$ defined by $I(\varphi_{\tau})({\bf x})=\psi({\bf x}, \tau)$  is an isomorphism of Hilbert spaces. 
\end{thm}
The map $I$ basically identifies each subspace $H_{\tau}$ with the usual space $L_{2}(\R^{3})$ of state functions on $\R^{3}$ considered at time $\tau$. The following result relates the dynamics on the family of subspaces $H_{\tau}$ and the usual space $L_{2}(\R^{3})$ of states of a spinless non-relativistic particle. 
\begin{thm}
\label{7}
Let ${\widehat h}=D+V({\bf x},t)$ be a Hamiltonian, such that $D$ is a differential operator in the spatial coordinates and $V$ is a function. Then the path $\varphi_{\tau}({\bf x},t)=\psi({\bf x},t)\delta(t-\tau)$ in $H$ satisfies the equation $\frac{d\varphi_{\tau}}{d\tau}=\left(-\frac{\partial}{\partial t}-i{\widehat h}\right)\varphi_{\tau}$ if and only if the function $\psi({\bf x},t)$ satisfies the Schr{\"o}dinger equation $\frac{\partial \psi({\bf x},t)}{\partial t}=-i{\widehat h}\psi({\bf x},t)$. 
At each point of the path $\varphi_{\tau}$ the components $-i{\widehat h}\varphi_{\tau}$, $-\frac{\partial \varphi_{\tau}}{\partial t}$ of the velocity vector $\frac{d \varphi_{\tau}}{d\tau}$ are orthogonal in the indefinite inner product, ${\widetilde H}$-inner product and the inner product on the space $H_{T}$ of the time co-moving representation. 
\end{thm}
In the theorem, the space $H_{T}$ of the time co-moving representation is defined by application of the isomorphism 
\begin{equation}
\label{timeco}
(\delta_{\Pi_{\tau}}f)({\bf x},t)=f({\bf x},t-\tau) 
\end{equation}
to the space $H$ (i.e., the representation is the map $\delta_{\Pi_{\tau}}\circ \Gamma$, where $\Gamma$ is the same as in theorem \ref{5}). The theorem claims that the ordinary Schr{\"o}dinger evolution can be recovered from the evolution $\varphi_{\tau}$ in the space $H$ of functions of four variables by projecting the path $\varphi_{\tau}$ onto the ``co-moving'' subspace $H_{\tau}$ identified via $I$ with ${\bf H}\approx L_{2}(\R^{3})$. While the component $-i{\widehat h}\varphi_{\tau}$ of the velocity describes the motion within the subspace $H_{\tau}$, the orthogonal (``vertical'') component $-\frac{\partial \varphi_{\tau}}{\partial t}$ of the velocity is due to the motion of the subspace $H_{\tau}$ itself. 

Remarks:
\begin{enumerate}
\item
Under integration in time the time variable gets replaced with the parameter $\tau$. 
In other words, for motions within the family $H_{\tau}$ the evolution parameter $\tau$ used to describe motions in the space $H$ of functions of four variables becomes identified with the usual time variable that appears in Schr{\"o}dinger equation. 
	\item 
The delta factor $\delta(t-\tau)$ in functions in $H_{\tau}$ removes integration in time and therefore eliminates the effect of interference in time that is present for more general elements of $H$. 
In fact, the norm of superposition
$\psi_{1}({\bf x},t)\delta(t-\tau)+\psi_{2}({\bf x},t)\delta(t-\tau)$
of functions in $H_{\tau}$
in either $H_{\eta}$, ${\widetilde H}$, or $H_{T}$-metrics is equal to
$\left\|\psi_{1}({\bf x},\tau)+\psi_{2}({\bf x},\tau)\right\|_{{\bf H}}$, which approximates the standard expression due to the relationship ${\bf H}\approx L_{2}(\R^{3})$. 
\item
The space $H$ was needed to identify Minkowski space with an isometrically embedded submanifold $M_{4} \subset H$. If this embedding is accepted, the reason for the delta factor $\delta(t-\tau)$ in the non-relativistic limit has a simple explanation.
In fact, elements of the space $H$ have the form $e^{-t^{2}}\varphi({\bf x}, t)$, where $\varphi$ is in the space ${\widetilde H}\approx L_{2}(\R^{4})$ (and the meaning of approximation is the same as in theorem \ref{1}). Likewise, the space $H_{T}$ of the time co-moving representation defined by Eq.(\ref{timeco}) consists of the functions $e^{-(t-\tau)^{2}}\varphi({\bf x},t)$, with $\varphi \in {\widetilde H}\approx L_{2}(\R^{4})$.
The variables ${\bf x},t$ enter symmetrically in the definition of $L_{2}(\R^{4})$, while the factor $e^{-(t-\tau)^{2}}$ breaks the symmetry between ${\bf x}$ and $t$ by making a typical element of $H_{T}$ well localized in the time variable. 
In a sufficiently small scale the factor $e^{-(t-\tau)^{2}}$ as a function of $t-\tau$ quickly falls off to almost zero and can be replaced with the delta function $\delta(t-\tau)$. This yields the set of functions in the family of spaces $H_{\tau}$ and by theorems \ref{6} and \ref{7} allows for the usual formalism of quantum mechanics. 
\item
Subspaces $H_{\tau}$ are not preserved under the maps $\delta_{\Pi}$ in theorem \ref{5}. In fact, $\delta_{\Pi}$ mixes space and time coordinates and therefore does not preserve the form $\varphi({\bf x},t)\delta(t-\tau)$ of elements of $H_{\tau}$ in general. This is not surprising because standard quantum mechanics is non-relativistic. 
However, to provide a valid foundation of the non-relativistic quantum mechanics these subspaces must be preserved under Galileo transformations.
A Galileo transformation $G$ yields the map $\delta_{G}:H \longrightarrow H'$ defined by $\delta_{G}f=f \circ G^{-1}$ for all $f \in H$.  This map 
transforms the state $\varphi({\bf x}, t)\delta(t-\tau)$ into the state $\varphi(A{\bf x}+{\bf v}t+{\bf b}, t+c)\delta(t+c-\tau)$, where $A$ is an orthogonal transformation, ${\bf v}$ and $ {\bf b}$ are $3$-vectors, and $c$ is a real number. Recall now that $\varphi$ is an element of the Hilbert space ${\bf H}$ with metric given by the kernel $e^{-\frac{1}{2}({\bf x}-{\bf y})^{2}}$. This kernel is obviously invariant under Galileo transformations so that the function  $\varphi(A{\bf x}+{\bf v}t+{\bf b}, t+c)$ is still an element of ${\bf H}$. One concludes that Galileo transformations yield isomorphisms between subspaces $H_{\tau}$ (and that the map $G \longrightarrow \delta_{G}$, where $\delta_{G}$ is considered as acting on ${\bf H}$ is a unitary representation of the Galileo group).
\item
The equation $\frac{d\varphi_{\tau}}{d\tau}=\left(-\frac{\partial}{\partial t}-i{\widehat h}\right)\varphi_{\tau}$ with usual Hamiltonian is a well known non-relativistic limit of the Stueckelberg-Schr{\"o}dinger equation in the theory of Stueckelberg Ref.\cite{Stu1} and Horwitz \& Piron Ref.\cite{HorPir}. This theory treats space and time symmetrically and predicts interference in time Refs.\cite{Hor},\cite{Hor2}.
The non-relativistic limit of Stueckelberg theory was investigated by Horwitz and Rotbart Ref.\cite{HorRot}. The approximate equality of the time variable $t$ with the evolution parameter $\tau$ obtained in Ref.\cite{HorRot} is consistent with the definition of $H_{\tau}$. 
\item
Newton and Wigner Ref.\cite{NW} argue that delta functions $\delta^{4}_{a}$ cannot represent spatially localized states in a relativistic theory. However, their derivation is based on the condition of orthogonality of a localized state and its spatial displacement, which is not valid in the proposed framework. Note that the delta function locality is present in the relativistic Stueckelberg theory, which is off-shell. If
the  Stueckelberg expectation value of the dynamical variable $\widehat {x^{\mu}}$ (the operator of multiplication by the variable $x^{\mu}$, $\mu=0,1,2,3$) is decomposed into a
direct integral over mass, then for each definite mass in the
integral, the Newton-Wigner operator (having Newton-Wigner localized states as eigenstates) emerges.  Locality is restored in
the result of the integral Refs.\cite{HorPir},\cite{HorRot}.

The covariant property of the states $\delta^{4}_{a}$ and the operator $\widehat {x^{\mu}}$ does not mean by itself that the found objects are physical. There are well known difficulties: (1) the wave packet $\delta^{3}_{\bf a}$ contains negative energy components; (2) if such a packet is allowed to evolve by the usual relativistic equations it will evolve out of the light cone Ref.\cite{Heg}. Although these difficulties are typical for relativistic on-shell wave equations and were understood within the Stueckelberg approach  Ref.\cite{HorPir}, they must be reexamined in the new setting. 
\end{enumerate}

\section{Generalizing the framework to curved space-time manifolds}

So far the discussion involved only the classical $3$-dimensional Euclidean space and the Minkowski space-time. If the approach is taken seriously, it becomes essential to check its validity for more general space-times $N$. 
It is also important to see whether the Hilbert space into which $N$ is embedded can be defined without specifying the manifold first. For manifolds without additional (pseudo-) Riemannian structure the issues are resolved by the following theorem. 
\begin{thm} 
\label{0}
Given an arbitrary real $n$-dimensional manifold $N$ there exists a Hilbert space $H_{\R^{n}}$ of continuous functions on $\R^{n}$, such that the set $M_{n}$ of all delta functions in the dual space $H_{\R^{n}}^{\ast}$ is an embedded submanifold of $H_{\R^{n}}^{\ast}$ diffeomorphic to $N$. 
\end{thm}
In essence, the theorem claims that an arbitrary $n$-dimensional manifold can be ``encoded'' into an appropriate Hilbert space of functions on $\R^{n}$. To get an idea of how to find the Hilbert space $H_{\R^{n}}$, especially when the topology of the manifold is not trivial, consider the case of a circle $S^{1}$. In this case the space $H_{\R}$ must be a Hilbert space of continuous functions on $\R$. To ensure that the image $M_{1}$ of the map $\omega: \R \longrightarrow H_{\R}^{\ast}$, $\omega(a)=\delta_{a}$ is a circle, one needs $\delta_{a}=\delta_{a+2\pi}$ for all $a \in \R$, which means that functions in $H_{\R}$ must be $2\pi$-periodic. To satisfy these conditions, consider the Sobolev space of continuous $2\pi$-periodic functions on $\R$ with the inner product $(f,g)=\int^{\pi}_{-\pi} \left(f(x)\overline{g}(x)+f'(x)\overline{g'}(x)\right)dx$. It is easy to check that the set of all delta functions in the dual space $H_{\R}^{\ast}$ with the induced topology is homeomorphic to the circle $S^{1}$.

A particular manifold in the theorem is encoded by fixing the {\em functional content} of the Hilbert space rather than fixing the domain of the functions. To put it differently, the manifold $M_{n}$ is ``made of'' functions and not points in the domain of the functions. 

The problem of isometric embeddings of Riemannian and pseudo-Riemannian manifolds is now handled by the following theorem.
\begin{thm}
\label{5new}
Let $N$ be a Riemannian or pseudo-Riemannian smooth manifold of dimension $n$. For any point $x\in N$ there is a neighborhood $W$ of $x$ in $N$ and a Hilbert or Krein space $H$ that contains delta functions (evaluation functionals) $\delta^{(n)}_{a}$ for all $a$ in an open set $U$ in $\R^{n}$ such that the set $M_{n}$ of all these delta functions is an embedded submanifold of $H$ isometric to $W$.
\end{thm}
\begin{proof}
It is known that an arbitrary smooth Riemannian or pseudo-Riemannian manifold $N$ of dimension $n$ admits an isometric embedding into the Euclidean or pseudo-Euclidean space $\R^{p}$ of a sufficiently large dimension $p\ge n$, $p=k+l$, where $(k,l)$ is the signature of the metric on $\R^{p}$.  Also, by an obvious generalization of theorems \ref{2} and \ref{3} the map $\Omega:\R^{p} \longrightarrow H_{p}$, $\Omega(A)=\delta^{(p)}_{A}$ is an isometric embedding of the space $\R^{p}$ into the Krein space $H_{p}$ defined via the inner product 
\begin{equation}
\label{innerMM}
(f,g)_{H_{p}}=\int e^{-\frac{1}{2}\left((X^{1}-Y^{1})^{2}+...+(X^{k}-Y^{k})^{2}\right)+\frac{1}{2}\left((X^{k+1}-Y^{k+1})^{2}+...+(X^{p}-Y^{p})^{2}\right)}f(X^{1},...,X^{p}){\overline g(Y^{1},...,Y^{p})} d^{p}X d^{p}Y
\end{equation}
with $d^{p}X=dX^{1}\cdots dX^{p}$, $d^{p}Y=dY^{1}\cdots dY^{p}$. Note that the analogues of $H_{ev}$, $H_{odd}$ in theorem \ref{2} are obtained here by representing an arbitrary function $f(X^{1},...,X^{k},X^{k+1},...,X^{p})$ as the sum of ``even''
\begin{equation}
\frac{1}{2}\left(f(X^{1},...,X^{k},X^{k+1},...,X^{p})+f(X^{1},...,X^{k},-X^{k+1},...,-X^{p})\right) 
\end{equation}
and ``odd''
\begin{equation}
\frac{1}{2}\left(f(X^{1},...,X^{k},X^{k+1},...,X^{p})-f(X^{1},...,X^{k},-X^{k+1},...,-X^{p})\right)
\end{equation}
components. Otherwise the proof mimics the one given in Ref.\cite{KryukovJMP}.

Let's form a Hilbert (Krein) subspace $H_{n}$ of $H_{p}$ in the following fashion. Let $x\in N$ be a point and let $X^{q}(u)$, $q=1,...,p$, $u \in U$, and $U \subset \R^{n}$ be functions describing the isometric embedding of a neighborhood $W \subset N$ of $x$ into $\R^{p}$. By permuting indices of the coordinates $X^{1}, ..., X^{p}$ and considering a smaller neighborhood $W$ if necessary one can always ensure that $u^{1},...,u^{n}$ are just the first $n$ of the coordinates. So, consider the set $H_{n}$ of all function(al)s in $H_{p}$ that have the form $\varphi(X^{1},...,X^{p})\delta(X^{n+1}-X^{n+1}(u))\cdots \delta(X^{p}-X^{p}(u))$, or more briefly $\varphi(X)\delta^{(p-n)}(X-X(u))$ with $u \in U$. Denoting the kernel of the metric in $H_{p}$, given by Eq.(\ref{innerMM}), by $k(X,Y)$, we have for the inner product of two such functionals:
\begin{equation}
\int k(X,Y) \varphi(X)\delta^{(p-n)}(X-X(u)){\overline \psi(Y)}\delta^{(p-n)}(Y-Y(v))d^{p}Xd^{p}Y,
\end{equation}
where $Y^{1}=v^{1},...,Y^{n}=v^{n}$. The delta functions remove integration with respect to $X^{n+1},..., X^{p}$ and $Y^{n+1},..., Y^{p}$, which gives
\begin{equation}
\label{newinner}
\int k(X(u),Y(v)) \varphi(X(u)){\overline \psi(Y(v))}d^{n}ud^{n}v,
\end{equation}
where $du=du^{1}\cdots du^{n}=dX^{1}\cdots dX^{n}$ and similarly for $dv$. The set $H_{n}$ is a closed subspace in $H_{p}$ so it is a Hilbert space. Expression (\ref{newinner}) shows that $H_{n}$ is isomorphic to the Hilbert space $H$ of all functions $\chi(u)=\varphi(X(u))$, $u \in U$ for which $\varphi(X)\delta^{(p-n)}(X-X(u))$ is in $H_{p}$, furnished with the inner product $(\chi, \rho)_{H}=\int k(X(u),Y(v)) \chi(u)\rho(v)d^{n}ud^{n}v$. 

Obviously, the functionals $\varphi(u)=\delta^{(n)}(u-a)$, $a \in U$ are in $H$. It remains to show that the metric induced on the set $M_{n}$ of all such functionals in $H$ is the given (pseudo) Riemannian metric on $N$. For this consider a curve $u^{\mu}=a^{\mu}(\tau)$ in $U$ and the corresponding curve $\varphi_{\tau}(u)=\delta^{(n)}(u-a(\tau))$ in $M_{n}$. For the squared $H$-norm of the velocity vector $d \delta^{(n)}(u-a(\tau))/d\tau$ we have
\begin{equation}
\label{10}
\int k(X(u),Y(v)) \frac{d \delta^{(n)}(u-a(\tau))}{d\tau}  \frac{d \delta^{(n)}(v-a(\tau))}{d\tau} d^{n}ud^{n}v.
\end{equation}
Simplifying this by the chain rule
\begin{equation}
\frac{d \delta^{(n)}(u-a(\tau))}{d\tau}=-\frac{\partial \delta^{(n)}(u-a(\tau))}{\partial u^{\mu}}\frac{d a^{\mu}(\tau)}{d\tau} 
\end{equation}
followed by integration by parts (see Ref.\cite{Kryukov3} for justification), one obtains the expression
\begin{equation}
\left.\frac{\partial^{2}k(X(u), Y(v))}{\partial u^{\mu} \partial v^{\nu}}\right|_{u=v=a(\tau)}\frac{d a^{\mu}(\tau)}{d\tau}\frac{d a^{\nu}(\tau)}{d\tau}.
\end{equation}
But
\begin{equation}
\frac{\partial^{2}k(X(u), Y(v))}{\partial u^{\mu} \partial v^{\nu}}=\frac{\partial^{2}k(X(u), Y(v))}{\partial X^{r} \partial Y^{s}}\frac{\partial X^{r}}{\partial u^{\mu}}\frac{\partial Y^{s}}{\partial v^{\nu}},
\end{equation}
and for the kernel $k(X,Y)$ given by Eq.(\ref{innerMM}) one also has
\begin{equation}
\left.\frac{\partial^{2}k(X(u), Y(v))}{\partial X^{r} \partial Y^{s}}\right|_{u=v}=\eta_{rs},
\end{equation}
where $\eta_{rs}$ are components of the indefinite (Minkowski-like) metric of signature $(k,l)$ on $\R^{p}$. So the squared norm of the velocity vector in Eq.(\ref{10}) is equal to
\begin{equation}
g_{\mu \nu}\frac{d a^{\mu}}{d\tau}\frac{d a^{\nu}}{d\tau}, 
\end{equation}
where 
\begin{equation}
\label{final}
g_{\mu \nu}=\eta_{rs} \left.\frac{\partial X^{r}}{\partial u^{\mu}}\frac{\partial Y^{s}}{\partial v^{\nu}}\right|_{u=v}
\end{equation}
are the components of the induced metric on $M_{n}$.

Recall now that the functions $X^{r}(u)$ describe the isometric embedding of $W \subset N$ into $\R^{p}$. In other words, components of the (pseudo-) Riemannian metric on $W$ are given by
\begin{equation}
{\widetilde g}_{\mu \nu}=\eta_{rs} \frac{\partial X^{r}}{\partial u^{\mu}}\frac{\partial X^{s}}{\partial u^{\nu}}.
\end{equation}
Since this expression coincides with Eq.(\ref{final}), the obtained embedding of $W$ into $H$ is isometric. This completes the proof.
\end{proof}
Several useful observations must be made.
\begin{enumerate}
\item
Theorem \ref{5new} makes it possible to extend the results of theorem \ref{3} to neighborhoods in arbitrary pseudo-Riemannian space-times. In this case the Poincar{\'e} group acting on Minkowski space-time is replaced by the group of diffeomorphisms of a particular neighborhood. This yields the following analogue of diagram (\ref{diagram})
\begin{equation}
\label{diagramB}
\begin{CD}
{\bf S}   @ >\Gamma>> H @ >\delta_{D}>> H'\\
@ AA\Omega A           @ AA \omega A      @ AA\omega A \\
W         @ >\gamma>> U @ >D>> U
\end{CD}
\end{equation}
Here $W$ is a neighborhood in curved space-time as defined in theorem \ref{5new}, $\gamma$ is a chart on $W$ and $U$ is the corresponding set in $\R^{4}$, $D$ is an arbitrary diffeomorphism of $U$ and $\delta_{D}$ is its extension to the space $H$ constructed in theorem \ref{5new}. As already mentioned in the remarks following theorem \ref{5}, the existence and uniqueness of the isomorphism $\delta_{D}$ and the space $H'$ can be proved as before. 
\item
Recall that the set $M_{4}$ is invariant under transformations $\delta_{\Pi}$, making it possible to ``separate'' special relativity from the Hilbert space framework. In the discussion that followed theorem \ref{7} it was verified that  the ``Galileo maps'' $\delta_{G}$ map subspaces $H_{\tau}$ onto themselves. This explains why the non-relativistic quantum mechanics could also be developed within a single Hilbert space of functions of three variables. 
Diagrams (\ref{diagram}), (\ref{diagramA}) provide us with a ``covariant'' extension of special relativity. Likewise, diagram (\ref{diagramB}) together with its active version yield a local geometric extension of general relativity. Those extensions are based on isomorphisms of separable Hilbert spaces. If such a scheme is adopted in physics, that would mean that specific functional realizations of the abstract Hilbert space ${\bf S}$, at least within the considered class of realizations, are not physical but rather are similar to various choices of coordinates on space-time. One may disregard this point by saying that the considered isomorphisms of Hilbert spaces of functions are direct analogues of well known changes in representation in quantum theory. However, unlike changes of representation that are simply passive changes in the description of physical reality, the transformations considered here can be realized actively. 
Active transformations are capable of creating a new physical reality. For  instance, rotation of a massive body can change the gravitational field created by it, while rotation of the coordinate system cannot. Inclusion of active transformations signifies then that the construction is not just formally mathematical, but is capable of affecting physics as well.
\item
If the discussed embedding of the classical space $\R^{3}$ into ${\bf H}$ as well as the embeddings of Minkowski space-time and local embeddings of arbitrary curved space-times into the corresponding Hilbert spaces  are taken seriously, then the linearity of quantum theory appears in a completely new light. In fact, the geometry of the abstract Hilbert space ${\bf S}$ and its realizations like $H$ is linear. It is the non-linearity of submanifolds $M_{3}$ and $M_{4}$ that seems to be responsible for the non-linear way in which classical world appears to us. 
By replacing the restricted, ``space-time based'' view of the world with its extension to the space ${\bf S}$ one can perhaps obtain a tool for reconciliation of  quantum theory and relativity.
\end{enumerate}

\section{Acknowledgments}

I am indebted to Larry Horwitz for his critical review of the results and an anonymous reviewer for useful recommendations and support. I would like to thank Malcolm Forster for numerous discussions and Kent Kromarek for help in improving the exposition. Part of this work was done at UW-Madison Department of Philosophy. I would like to thank the faculty of that department for their hospitality. This work was supported by the NSF grant SES-0849418.


\end{document}